\documentclass[twocolumn,aps,pra,showpacs,groupedaddress]{revtex4-1}
\usepackage{epsfig,amsmath,amssymb,mathrsfs,mathtools,amscd}
\usepackage[linkcolor=red]{hyperref}

\usepackage{cleveref}
\usepackage{soul}
\usepackage{enumerate}
\usepackage{amsmath} 
\usepackage{amsthm}
\usepackage{verbatim} 
\usepackage{mathtools} 
\usepackage{float} 
\usepackage{tikz}
\usepackage{stackengine}

\newtheorem{lemma}{Lemma}

\newtheorem{proposition}{Proposition}

\mathchardef\minus="002D

\def\<{\langle}
\def\>{\rangle}
 \def\ket#1{| #1 \rangle}
\def\bra#1{\langle #1 |}
\def\ketbra#1#2{| #1 \rangle \langle#2 |}
\def\braket#1#2{\langle #1 | #2 \rangle}

\DeclareMathOperator{\spn}{span}
\DeclareMathOperator{\supp}{supp}

\newcommand\Item[1][]{%
  \ifx\relax#1\relax  \item \else \item[#1] \fi
  \abovedisplayskip=0pt\abovedisplayshortskip=0pt~\vspace*{-\baselineskip}}

\begin{document}
\title{A massless interacting Fermionic Cellular Automaton
  exhibiting bound states}

\author{Edoardo
  \surname{Centofanti}} \email[]{edoardo.centofanti01@universitadipavia.it}
\affiliation{Dipartimento di Matematica dell'Universit\`a di Pavia, via
  Ferrata 5, 27100 Pavia} 
\author{Paolo
  \surname{Perinotti}} \email[]{paolo.perinotti@unipv.it}
\affiliation{Dipartimento di Fisica dell'Universit\`a di Pavia, via
  Bassi 6, 27100 Pavia} \affiliation{Istituto Nazionale di Fisica
  Nucleare, Gruppo IV, via Bassi 6, 27100 Pavia} 
\author{Alessandro 
  \surname{Bisio}} \email[]{alessandro.bisio@unipv.it}
\affiliation{Dipartimento di Fisica dell'Universit\`a di Pavia, via
  Bassi 6, 27100 Pavia} \affiliation{Istituto Nazionale di Fisica
  Nucleare, Gruppo IV, via Bassi 6, 27100 Pavia} 

\begin{abstract} 
  We present a Fermionic Cellular Automaton model
  which describes massless Dirac fermion in $1+1$
  dimension coupled with local, number preserving
  interaction.  The diagonalization of the two
  particle sector shows that specific values of
  the total momentum and of the coupling constant
  allows for the formation of bound states.
  Furthermore, we present a classification of the
  local number-preserving interactions that are
  invariant under the isotropy group of the
  cellular automaton which simulates the Weyl
  equation.
\end{abstract} 
\maketitle

\section{Introduction}

Quantum cellular automata
(QCAs)\cite{Farrelly_2020,arrighi2019overview,schumacher2004reversible,gross2012index}
are the most general unitary dynamics of a lattice
of quantum systems which is discrete in time and
local, namely the speed of information
propagation is bounded.  The idea of QCAs can be
traced back to the seminal work of Feynman
\cite{feynman1982simulating} where QCAs were
introduced as quantum simulators.  Since then,
QCAs have been considered as a paradigm for
quantum
computation\cite{watrous1995one,PhysRevA.72.022301,PhysRevLett.102.180501,PhysRevA.81.042330}
and have been applied to the study of many bodies
quantum systems \cite{cirac2017matrix,PhysRevX.6.041070,haah2023nontrivial,PhysRevLett.125.190402,Zimboras2022doescausaldynamics,Hillberry_2021}.

Nevertheless, quantum simulation still is a
major application of QCAs, in particular as
discretized quantum field theories.
Many authors \cite{Bialynicki_Birula_1994,meyer1996quantum,PhysRevA.73.054302,Yepez:2006p4406,PhysRevA.90.062106,Arrighi_2014,Bisio2015244,bisio2014quantum,PhysRevA.97.062111} studied the simulation of non-interacting quantum
field theories with  discrete-time Quantum Walks
(QWs) \cite{ambainis2001one,portugal2013quantum},
which are the single-particle restriction of
Bosonic QCAs or Fermionic Cellular Automata 
(FCA)---a variation of QCAs where the cells 
correspond to arrays of local Fermionic modes~\cite{Bravyi2002210}. 
Indeed, a QW describes the most general discrete-time
evolution of a single particle on a lattice which
is unitary and local, i.e. at each step the
particle can move at most by a bounded number of
lattice sites. More recently, the focus shifted
towards to
models in which the particles interact with an
external potential \cite{PhysRevA.88.032301,PhysRevA.88.042301,PhysRevA.93.052301,Jay:2020aa}
and towards interacting field theories
\cite{bisio2018thirring,eon2022relativistic}.

Intuitively, we expect that a QCA/FCA which simulates
(or recovers) a given continuous dynamics is such
that, if we restrict to sufficiently smeared 
states which cannot probe the discreteness of the
lattice, we cannot tell the evolution apart from
the dynamics of a field on a continuous space.
For the non-interacting case this behaviour is
rather well understood. In the limit of small
masses and momentum, where we can neglect powers of  $m,k$
beyond the linear terms, and interpolating 
discrete time steps with a continuous time direction, 
the evolution of the cellular automaton is described 
by a relativistic wave equation. The interacting case, 
still rather unexplored, is significantly different. 
The intrinsic discreteness of cellular automata may 
produce phenomenological features with no counterpart 
in the continuum and that do not vanish in the large
scale limit, like a different set of bound states 
and scattering processes that conserve
energy only up to integer multiples of an additive constant 
(indeed, if time is discrete, energy is defined only modulo to 
the addition of integer multiples of $2\pi$, and can thus be 
reduced in the interval $[-\pi,\pi]$) 
\cite{bisio2018thirring,PhysRevLett.126.250503}.

In this work we classify the family of local
and number preserving interactions which are
invariant under the isotropy group of the Weyl
automaton, namely the fermionic cellular automaton
which recovers the Weyl equation in 3+1
dimensions.

Then we study an interacting $1+1$-dimensional FCA
given by the composition of a free evolution,
which is modeled after the massless Dirac equation
\cite{PhysRevA.90.062106} followed by a quartic,
number preserving local interaction. We chose this
interaction among the isotropic 
ones that we classified.

We will study the
evolution of
the two particle
sector by providing the spectral resolution of the
unitary evolution corresponding to a single step
of the FCA. 
Our analysis will show that, even if the free
theory is gapless and the interaction is a compact
perturbation, bound states appear for critical values 
of the total momentum and of the coupling constant. 

In section \ref{sec:2}
we review the technical tools employed for
QCAs and QWs and provide a mathematical description
of the Dirac automaton.

In section \ref{sec:3} we provide the classification of all
the suitable local, isotropic and number preserving interactions 
for our automaton, in particular for the 
two particle sector and we describe the interaction 
considered among the ones derived and its properties, 
then we derive the scattering and the
bounded states of the system. In appendices
\ref{sec:compl-solut} and \ref{sec:study-r_pm-} we
check the completeness of the solution set.

\section{Fermionic Cellular Automata}\label{sec:2}

A quantum cellular automaton
(QCA)\cite{schumacher2004reversible} describes the
single step unitary evolution $\mathcal{U}$ of a
set $\Gamma$ of cells each of which represents a
quantum system. Usually,  every cell
consists in a $d$-level system. However, we are
interested in Fermionic Cellular Automata (FCA)s,
where every
cell corresponds to $s$ local Fermionic
modes~\cite{kitaev1997quantum,d2014feynman}.
represented by the field operators $\psi_{x,a}$,
$a=1, \dots, s$ obeying the canonical
anticommutation relations:
\begin{align}
  \label{eq:2}
\{\psi_{x,a},\psi_{x',a'}\}=0, \quad \{\psi_{x,a},\psi_{x',a'}^\dagger\}=\delta_{x,x'}\delta_{a,a'}.
\end{align}
The $N$-particles states can be described by
introducing the Fock space representation
$\ket{(x_1,a_1),\ldots,(x_N,a_N)}\coloneqq
\psi^\dag_{x_1,a_1}\ldots
\psi^\dag_{x_N,a_N}\ket\Omega$, where $\ket\Omega$
is the {\em vacuum state}. The state $\ket\Omega$
is defined by the identity
$\psi_{x,a}\ket\Omega =0$ for all $i$ and it is
the state with no particles.  Here, we chose the
representation in which the vacuum state is the
one with no \emph{localized} excitations, whereas
in quantum field theory the vacuum state is
usually the ground state of a free Hamiltonian.

\subsection{Linear FCA with symmetries: Weyl and
  Dirac automata}
 In the particular case
of a free, i.e.~non-interacting, evolution, the
FCA is linear in the field operators, namely
$\mathcal{U}( \psi^\dag_{x,a})=\sum_{y,b}U^{y,b}_{x,a}
\psi_{y,b}$
for some complex coefficents $U_{x,a}^{y,b}$ and
the number of excitations
(i.e. particles) is
conserved. Therefore, the
dynamics is
completely specified by a \emph{quantum walk}
on the one-particle
sector
\begin{align}
  \label{eq:QW}
  \begin{aligned}
     &\ket{\phi(t+1)} =  \mathbf{U}
     \ket{\phi(t)},\\
     & \mathbf{U} \in \mathcal{L}(\mathcal{H}_1),
     \; \mathcal{H}_1\coloneqq \spn \{\ket{x,a}\}
\equiv \ell^2(\Gamma) \otimes \mathbb{C}^s,\\
     &\mathbf{U} \coloneqq
     \sum_{x,y,a,b} U_{x,a}^{y,b} \,
     \ketbra{x}{y}\otimes\ketbra{a}{b} ,
  \end{aligned}
\end{align}
where we use the isomorphism $\ket{x,a} \leftrightarrow
\ket{a}\ket{x}$.
Accordingly, the $n$-particle sector of a
linear FCA $\mathcal{U}$ is described
by the total antisymmetric subspace
of the $n$-particle
quantum walk $\mathbf{U} ^{\otimes{n}}$.
Therefore, a linear
FCA can be regarded as the
second-quantization of a Quantum Walk.

\begin{figure}[t]
    \includegraphics
   [width=0.7\columnwidth]{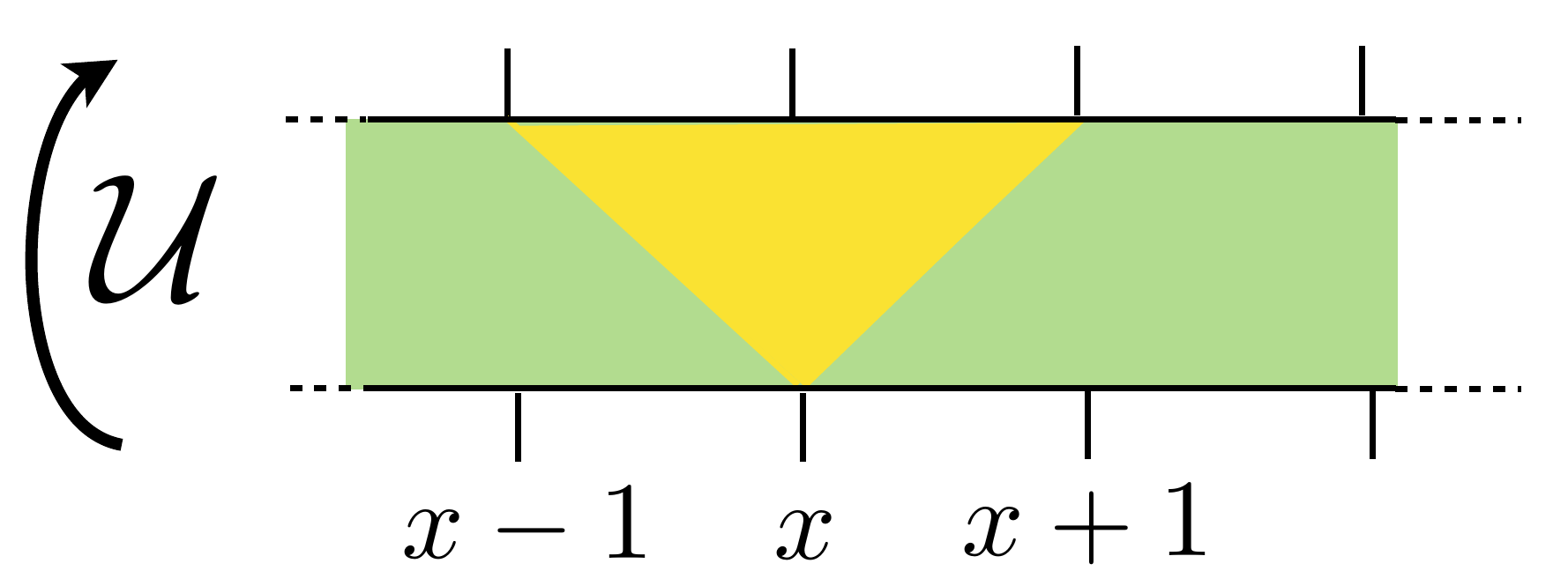}
   \caption{(Colors online) In a one-dimensional cellular
     automaton each cell is labelled by an integer
     number $x$ and the one-step time evolution is
     described by a unitary operator
     $\mathcal {U}$. In a nearest-neighbor
     one-dimensional QCA the state of the cell $x$
     at time $t+1$ depends only on the state of
     the cells $x-1$, $x$ and $x+1$ at time $t$.
     Because of the locality of the evolution,
     information cannot propagate
     at an arbitrary speed: physical interventions on
     cell $x$  can
     only affect the state of the cells lying in the future light
     cone of $x$ (in yellow). }
   \label{fig:1DQCA}
 \end{figure}

The FCA evolution must also be \emph{local}, namely,
for a fixed $x \in \Gamma$ $U_{x,a}^{y,b} \neq 0$
only if for finitely many $y \in \Gamma$, which
are called the \emph{neighbors} of $x$ (see
Fig.\ref{fig:1DQCA}).  From this locality
assumption, we have that the set $G$ is naturally
endowed with a graph structure, where
$x \in \Gamma$ are the vertices and there is an
edge between a cell and its neighbors.

In the
most general case, the update rule may change from
cell to cell \cite{gross2012index}. However, it is
reasonable to assume a \emph{homogeneous}
evolution, with the same upate rule for
every cell.  More precisely, the homogeniety
assumptions requires \cite{PhysRevA.90.062106,PhysRevA.100.012105} that the graph $\Gamma$ is
the Cayley graph of a group $G$ that can be
presented as $G = \langle S |R\rangle$
where $S$ is a
finite set of generators
and  $R$ is a finite set
of relators.
Each edge connecting a cell $x$ with
its neighbors corresponds to an element of
$S$ 
and elements of  $R$ correspond to closed path on
the graph.
If a linear FCA is local and homogeneous 
the unitary operator $\mathbf{U}$  of Equation
\eqref{eq:7} can be written as follows:
\begin{align}
  \label{eq:43}
  \mathbf{U} := \sum_{h \in S} \mathbf{T}_h
  \otimes \mathbf{U}_h
\end{align}
where 
$\mathbf{T}$ is the right regular representation
of $G$ on $\ell^2(G)$  acting as $\mathbf{T}_h
\ket{x} := \ket{xh^{-1}}$
and $\mathbf{U}_h \in \mathcal{L}(\mathbb{C}^s)$.

Let us now review a notion of \emph{isotropy}
\cite{PhysRevA.96.062101} for a FCA as in Equation
\eqref{eq:43}.  Let us consider a decomposition
$S = S_+ \cup s_- \cup \{e\}$ of the set of
generators ($e$ denotes the identity on $G$) and
let $L$ be a group of graph automorphisms that is
transitive over $S_+$ and
$\mathbf{V} : L \to \mathcal{L}(\mathbb{C}^s)$ be
a faithful projective representation of $L$.  Then,
we say that $\mathbf{U}$ is isotropic if the
following covariance condition holds:
\begin{align}
  \label{eq:41}
  \mathbf{U} =
  \sum_{h \in S} \mathbf{T}_h
  \otimes \mathbf{U}_h =
  \sum_{h \in S} \mathbf{T}_h
  \otimes \mathbf{V}_l \mathbf{U}_h
  \mathbf{V}^\dag_l, \quad \forall l \in L.
\end{align}
Generally, the same FCA on the Cayley graph
corresponding to the presentation $G = \langle S |
R \rangle$ might
satisfy isotropy for one or more
choices of the set $S_+$, group $L$ and the
representation $\mathbf{U}$. For a given $S$,
different choices of $S_+$ correspond to different
orientations of some edges over the same 
graph.

The Cayley graph $\Gamma$ can be endowed with the
word metric: the distance between $x$ and $y$ is
equal to the lenght of the shortest path
connecting $x$ and $y$.  If the dynamics of a free
fermionic field on flat spacetime is expected to
emerge as a large scale description of the FCA,
then $\Gamma$ must be quasi isometrically
embeddable \footnote{Given two metric spaces
  $(M_1,d_1)$ and $(M_2,d_2)$, with $d_1$ and
  $d_2$ the metric of the two spaces, a map
  $f:(M_1,d_1)\rightarrow (M_2,d_2)$ is a
  quasi-isometry \cite{de2000topics} if there
  exist constants $A\geq 1$, $B,C\geq 0$, such
  that $\forall g_1,g_2\in M_1$ one has
  $ d_1(g_1,g_2)/A-B\leq d_2(f(g_1),f(g_2))\leq A
  d_1(g_1,g_2)+B, $ and $\forall m\in M_2$ there
  exists $g\in M_1$ such that
  $d_2(f(g),m)\leq C $. Intuitively, if there
  exists a quasi isometry between $(M_1,d_1)$ and
  $(M_2,d_2)$, then the large scale geometry of
  the two spaces are the same.  } in the Euclidean
space $\mathbb{R}^3$.  A well known result in
geometric group theory \cite{gromov1984infinite}
states that a group $G$ is quasi-isometrically
embeddable in $\mathbb{R}^n$ space if and only if
it has a subgroup isomorphic to $\mathbb{Z}^n$ of
finite index.

\begin{figure}[t]
    \includegraphics
   [width=\columnwidth]{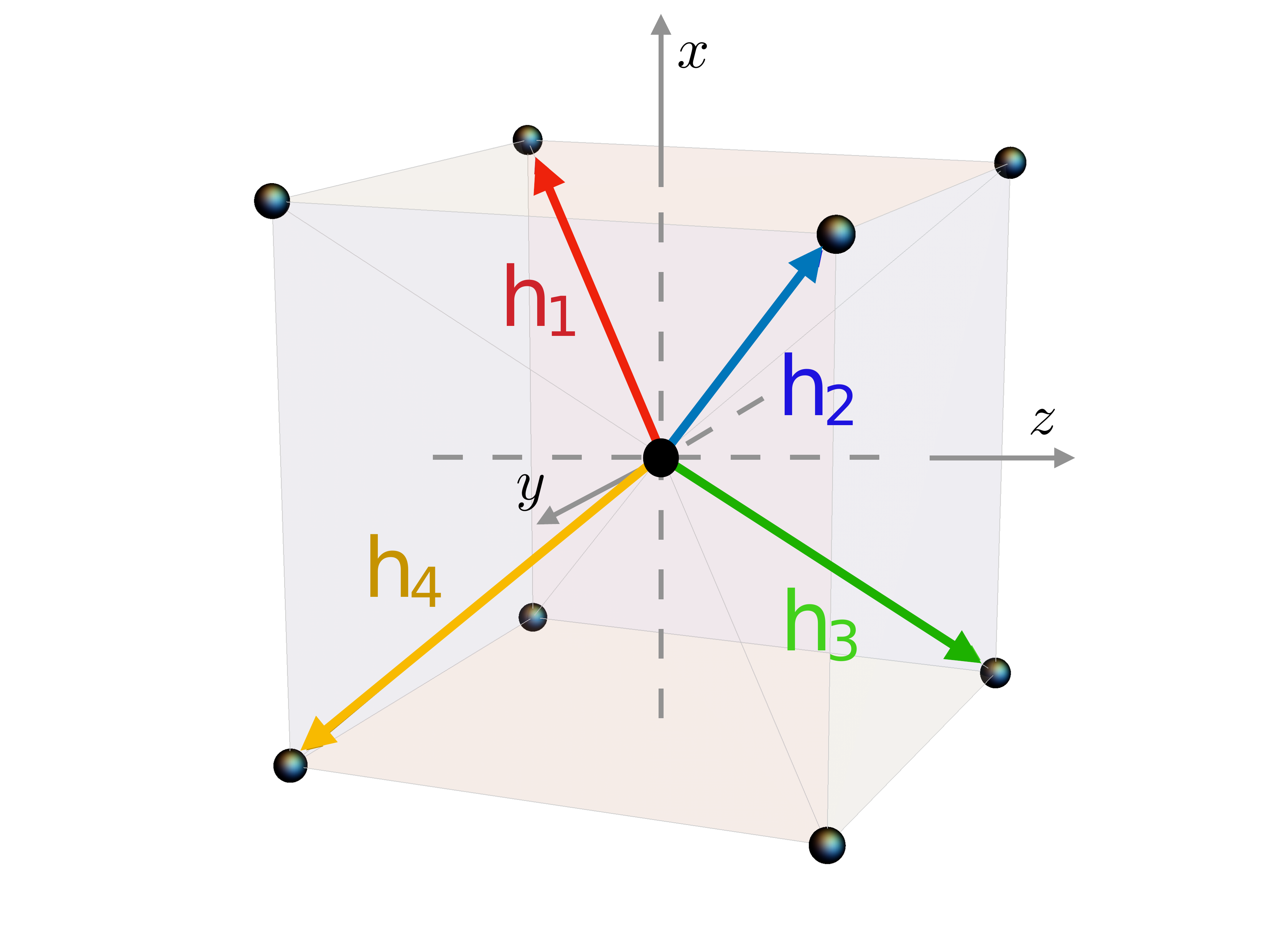}
   \caption{The body-centered cubic (BCC) lattice
     is the Cayley graph of $\mathbb{Z}^3$
     corresponding to the presentation
     $\mathbb{Z}^3 = \langle h_1, h_2, h_3, h_4| h_1+h_2+h_3+h_4 \rangle$ (the abelianity relators are omitted).
     The isotropy group are the $\pi$ rotation
     around the $x,y$ and $z$ axis, i.e.
     $L:= \{ I, R(\pi,x), R(\pi,y), R(\pi,z) \}$.}
   \label{fig:BCC}
 \end{figure}

Let us now restrict ourselves
to the case $G = \mathbb{Z}^3$ and to FCAs with two
femionic degrees of freedom per cell ($s=2$).
Then, the only graph which admits an isotropic linear
FCA  is the BCC lattice \cite{PhysRevA.96.062101} 
and we only have the following two isotropic
linear FCAs (see Fig.\ref{fig:BCC}):

\begin{align}
  \label{eq:WeylWalk}
&  \mathbf{W}^{(\pm)} := \sum_{h \in S}
\mathbf{T}_h \otimes \mathbf{W}^{(\pm)}_{h}, 
  \\
&\begin{aligned}
  \mathbf{W}^{(\pm)}_{h_1}&:= 
  \begin{pmatrix}
    \eta^{\pm}&0\\
    \eta^{\pm}&0
  \end{pmatrix}, &
  \mathbf{W}^{(\pm)}_{-h_1}&:= 
  \begin{pmatrix}
    0&-\eta^{\mp}\\
    0&\eta^{\mp}
  \end{pmatrix},\\
\mathbf{W}^{(\pm)}_{\pm h_2} &:= \sigma_x \mathbf{W}^{(\pm)}_{\pm h_1}
\sigma_x,  &
\mathbf{W}^{(\pm)}_{\pm h_3} &:= \sigma_y \mathbf{W}^{(\pm)}_{h_1}
\sigma_y , 
\\
\mathbf{W}^{(\pm)}_{\pm h_4} &:= \sigma_z \mathbf{W}^{(\pm)}_{h_1}
\sigma_z , & \eta^{\pm} &:= \frac{1\pm i}{4}  .
\end{aligned}
\end{align}
with $S_+ = \{h_1,h_2,h_3,h_4\}$ and the
nontrivial relator $h_1+h_2+h_3+h_4=0$.
Given an orthonormal reference frame, the
generators can be chosen as follows (see Fig.~\ref{fig:BCC})
\begin{equation}
  \begin{split}
    &h_1=\frac 1{\sqrt3}
    \begin{pmatrix}
      1\\
      -1\\
      -1
    \end{pmatrix},\ 
    h_2=\frac 1{\sqrt3}
    \begin{pmatrix}
      1\\
      1\\
      1
    \end{pmatrix},\\
    &h_3=\frac 1{\sqrt3}
    \begin{pmatrix}
      -1\\
      -1\\
      1
    \end{pmatrix},\ 
    h_4=\frac 1{\sqrt3}
    \begin{pmatrix}
      -1\\
      1\\
      -1
    \end{pmatrix}.
  \end{split}
  \label{eq:vers}
\end{equation}
The isotropy group of $W$ is the subgroup of
$SO(3)$ given by the the $\pi$ rotation around the
$x,y$ and $z$ axis, and
we have the usual spin-$1/2$ representation on the
internal degrees of freedom, i.e.
\begin{align}
  \label{eq:52}
  &L:= \{ I, R(\pi,x), R(\pi,y), R(\pi,z), \}\\
 & \begin{aligned}
  &\mathbf{V}^{W}_{I} = I, && 
\mathbf{V}^{W}_{R(\pi,x)} = i \sigma_x, \\
&\mathbf{V}^{W}_{R(\pi,y)} = i \sigma_y, &&
\mathbf{V}^{W}_{R(\pi,z)} = i \sigma_z.    
  \end{aligned}
\end{align}
For the abelian group $\mathbb{Z}^3$, the
homogeneity assumption is the usual translation
invariance. It is therefore convenient to
introduce the
the Fourier transform of the local modes
\begin{align}
  \label{eq:44}
  \begin{aligned}
\psi^\dag_{k,a} &\coloneqq
\sum_x e^{-ik\cdot x} \psi^\dag_{x,a}, \quad 
           k \in
  B,    
  \end{aligned}
\end{align}
where $B$ denotes the Brillouin zone of the BCC
lattice.  Then, a generic one-particle state can
be written as
$\ket{\phi} = \int_{B} dk \sum_{a}\tilde{c}_{a,k}
\ket{a}\ket{k}$ with
$\ket{k}\coloneqq\sum_{x}e^{-ik\cdot x}\ket{x}$
and the unitary operator of Equation
\eqref{eq:WeylWalk} becomes \cite{PhysRevA.90.062106}
\begin{align}
  \label{eq:42}
&\mathbf{W}^{(\pm)} =  \int_B dk \ketbra{k}{k}
  \otimes \mathbf{W}^{(\pm)} (k)\\
  &\begin{aligned}
    \mathbf{W}^{(\pm)} (k) & = d^{(\pm)}(k) I -
                                    i a^{(\pm)}_x(k) \sigma_x  + \\
 & \pm
  i a^{(\pm)}_y(k) \sigma_y +
  -i a^{(\pm)}_z(k) \sigma_z    
  \end{aligned}
  \\
  &\begin{aligned}
  &d^{(\pm)}(k):=c_xc_yc_z\mp s_xs_ys_z \\
  &a^{(\pm)}_x(k)  := s_xc_yc_z  \pm c_xs_ys_z \\
  &a^{(\pm)}_y(k) :=c_xs_yc_z  \mp  s_xc_ys_z \\
  &a^{(\pm)}_z(k) :=c_xc_ys_z \pm s_xs_yc_z.   
\end{aligned}\\
  &c_i := \cos\tfrac{k_i}{\sqrt3} \quad s_i := \cos\tfrac{k_i}{\sqrt3}
\end{align}
One could verify \cite{PhysRevA.90.062106} that,
in the limit $k\to 0$ we have
$ \mathbf{W}^{(\pm)} (k) \approx \exp(-i
H^{\pm}_W(k))$ where $H^{+}_W(k)$ is the
Hamiltonian of the right handed Weyl equation and
$H^{-}_W(k)$ is the Hamiltonian of the left handed
Weyl equation (up to a change of coordinates).
Therefore, the FCA in Equation \eqref{eq:WeylWalk}
is called Weyl cellular automaton.  The Dirac
automaton arises \cite{PhysRevA.90.062106} from the
local coupling of two Weyl automata in a
direct sum:
\begin{align}
  \label{eq:50}
  &\mathbf{D} :=  \int_B dk \ketbra{k}{k}
                  \otimes \mathbf{D} (k)\\
  &\begin{aligned}
  &\mathbf{D}(k) :=
                  \begin{pmatrix}
                  n \mathbf{W} (k)  & im I
                  \\
                  im I & n  \mathbf{W}^{\dag} (k) 
                  \end{pmatrix} \\
                 & n^2+m^2 =1,   
               \end{aligned} 
\end{align}
where we omitted the $\pm$ superscript in order to
loghten the notation.
It is easy to verify that in the $k \to 0$ and
$m\to 0$ limit, the Dirac automaton of
Equation~\eqref{eq:50} recovers the Dirac
equation, with the parameter $m$ playing the role
of the rest mass.

The Dirac automaton inherit the same isotropy group $L$
of the Weyl auomata (see Equation \eqref{eq:52})
with the corresponding direct sum representation
on the internal degrees of freedom:
\begin{align}
  \label{eq:53}
  \begin{aligned}
  &\mathbf{V}^{D}_{I} = I, && 
\mathbf{V}^{D}_{R(\pi,x)} = - \sigma_x\oplus  \sigma_x, \\
&\mathbf{V}^{D}_{R(\pi,y)} = - \sigma_y \oplus \sigma_y, &&
\mathbf{V}^{D}_{R(\pi,z)} = - \sigma_z \oplus \sigma_y.          
  \end{aligned}
\end{align}

\section{Dirac automaton with isotropic local interactions}\label{sec:3}
In the previous section we focused on linear FCAs
QCAs.  Describing an interaction, instead,
involves non-linear evolutions of the field
operators.  One could interpret a local
interaction term as a local change of basis which
models the arbitrariness of the identification of
local bases at subsequent {\em discrete} time
steps, as in a discrete version of a gauge theory.
At each site $x$ of the graph we have the action
of a local gate $\exp{(-i J_x)}$ where $J_x$ is some
Hermitian operator localized at site $x$.  Then,
the interaction step can be written as follows:
\begin{align}
 & \mathcal{J} \big( \psi_{x,a} \big) \coloneqq
\exp ( -i J)
  \psi_{x,a}\exp \big( i J\big), \quad
  J \coloneqq \sum_x J_x.
\end{align}

It is natural to require that the interaction term
has the same symmetries of the free
evolution. For a localized interaction,
transaltion invariance trivially implies that the
interacting gate is the same at all sites.
The isotropy assumption requires that
the local interaction
must commute
with the representation of the isotropy group on
the internal degrees of freedom.
For the Dirac automaton of Equation
\eqref{eq:50}, the isotropy assumption implies that
the interaction term commutes with 
 the discrete rotations
of Equation  \eqref{eq:53}.

In this paper, we focus our analysis on the simplest
case of number-preserving interactions.
The most general local, isotropic and number
preserving interaction can be written as follows:
\begin{align}
  \label{eq:class}
 & J_x := J_{x}^{q} + J_{x}^{e}+ J_{x}^{o}, \\
 & \begin{aligned}
  \label{eq:3}
  J_{x}^{q} :=&
  \lambda_1
  (n_{x,3} \psi_{x,4}^\dag\psi_{x,2}  +
  n_{x,4} \psi^\dag_{x,3} \psi_{x,1}) + h.c.+\\
  &\lambda_2
  (n_{x,1} \psi_{x,2}^\dag  \psi_{x,4} +
  n_{x,2} \psi_{x,1}^\dag  \psi_{x,3}) + h.c. + 
  \\
 &\lambda_3
 (\psi^\dag_{x,1} \psi^\dag_{x,2}   \psi_{x,3}
  \psi_{x,4}) + h.c. + \\
   & \lambda_4
(\psi^\dag_{x,2} \psi^\dag_{x,4}   \psi_{x,1}
  \psi_{x,3} + h.c.)+ \\
  &\lambda_5
  (\psi^\dag_{x,2} \psi^\dag_{x,3}   \psi_{x,1}
  \psi_{x,4} + h.c.) +\\
  &\lambda_6
  (n_{x,1}n_{x,3} + n_{x,2}n_{x,4}) + \\
  &\lambda_7
  (n_{x,1}n_{x,4} + n_{x,2}n_{x,3})+\\
  &\lambda_8
  n_{x,1}n_{x,2}+
  \lambda_9
  n_{x,3}n_{x,4}   , 
\end{aligned}  \\
  &
\begin{aligned}
    J_{x}^{e} := 
&    \xi_1
    (n_{x,1}n_{x,3}\psi^\dag_{x,4}\psi _{x,2} + \\
    & +n_{x,2}n_{x,4}\psi^\dag_{x,3}\psi _{x,1}) +h.c +
    \\
  &  \xi_2
    (n_{x,1}+n_{x,2})n_{x,3}n_{x,4}+
    \\
   & \xi_3
    n_{x,1}n_{x,2}(n_{x,3}+n_{x,4}) ,
  \end{aligned} \\
&  J_x^{o} := \chi   n_{x,1}n_{x,2}n_{x,3}n_{x,4}.
\end{align}
where we introduced the number operators
$n_{x,i} := \psi^\dag_{x,i}\psi_{x,i}$.
The proof of the classification has been performed
through the symbolic calculation python module
\texttt{sympy} and the code used is available on
GitHub.

We now procced with the phenomenological analysis
of an interacting FCA. We will restrict to the
easiest non-trivial case, namely two massless
4-spinors in one dimension.  We thus consider a FCA
model $\mathcal A$ of the form
\begin{align}
  \label{eq:1}
  \mathcal{A} \coloneqq\mathcal{J} \circ \mathcal{F},
\end{align}
where $\mathcal{F}$ is the Dirac FCA
in the one dimensional case for $m=0$, i.e.
\begin{align}
  \label{eq:7}
  &\mathbf{F} :=  \int_{\pi}^{\pi} dk \ketbra{k}{k}
                  \otimes \mathbf{F} (k)\\
  &\begin{aligned}
  &\mathbf{F}(k) :=
                  \begin{pmatrix}
                 e^{-ik}  &0 &0 &0 \\
                 0 & e^{ik} &0 &0 \\
                 0  &0 & e^{ik} &0 \\
                  0  &0 &0 & e^{-ik} 
                  \end{pmatrix} 
               \end{aligned}, 
\end{align}
or, in the position representation 
\begin{align}
  \begin{aligned}
&      \mathbf{F} =
  \mathbf{T} \otimes
  (\ketbra{1}{1} + \ketbra{4}{4}) +
  \mathbf{T}^\dag
  \otimes (\ketbra{2}{2} + \ketbra{3}{3}), \\
  & \mathbf{T} \ket{x} := \ket{x+1},
  \end{aligned}
\end{align}
and $\mathcal{J}$ is the interacting step
(see Fig. \ref{fig:interactingQCA}).
Since we are
dealing with a Fermionic system, the two particle
Hilbert space is given by the antisymmetric
subspace $\mathcal{H}_-$ of the tensor product
$\mathcal{H}_1\otimes \mathcal{H}_1$ of two copies of
the one particle space $\mathcal{H}_1 =
\spn \big \{ \ket{a}
  \ket{x}\big \} \equiv \mathbb{C}^4 \otimes \ell^2(\mathbb{Z})$ , where
$ a \in
  \{1,2,3,4\}$ and $ x \in
  \mathbb{Z} $.
  Namely, we have
  \begin{align}
\label{eq:14}    
\begin{aligned}
      \mathcal{H}_- \coloneqq \supp  \mathbf{P}_- ,\quad
    \mathbf{P}_- \coloneqq \frac{1}{2}\big(\mathbf{I} -
    \mathbf{S}\big),
    \end{aligned}
  \end{align}
where $   \mathbf{I} $ is the identity on
$\mathcal{H}_1 \otimes \mathcal{H}_1$ and
$\mathbf{S}$ is the swap operator
$\mathbf{S}
    \ket{a_1}\ket{a_2}
  \ket{x_1}
    \ket{x_2} =
    \ket{a_2}\ket{a_1}
  \ket{x_2}
  \ket{x_1}$, where 
  \begin{align}
  \label{eq:12}
  \begin{aligned}
  &\mathcal{H}_1 \otimes \mathcal{H}_1 = 
  \spn \big \{ \ket{a_1}\ket{a_2}
  \ket{x_1}\ket{x_2} \big \}, \\
  &\, a_1,a_2 \in
  \{1,2,3,4\} , \quad  x_1,x_2 \in
  \mathbb{Z}  .
  \end{aligned}
  \end{align}
Let us consider the two-particle sector of
$\mathcal{F}$.
We have
\begin{align}
  \label{eq:5}
  \begin{aligned}
  \mathbf{F} \otimes \mathbf{F} =&
  \mathbf{T}\otimes \mathbf{T} \otimes
\sum_{i,j \in \{1,4\}}  \ketbra{i}{i}\otimes
  \ketbra{j}{j} +\\
  & +\mathbf{T}\otimes \mathbf{T}^\dag \otimes
\sum_{\substack{i \in \{1,4\}\\ j\in \{2,3\}}}  \ketbra{i}{i}\otimes
  \ketbra{j}{j} \\
  &+\mathbf{T}^\dag \otimes \mathbf{T} \otimes
\sum_{\substack{i\in \{2,3\} \\ j \in \{1,4\}  }}  \ketbra{i}{i}\otimes
  \ketbra{j}{j} + \\
  &+ \mathbf{T}^\dag \otimes \mathbf{T}^\dag \otimes
\sum_{\substack{i\in \{2,3\} \\ j \in \{2,3\}  }}  \ketbra{i}{i}\otimes
  \ketbra{j}{j} .    
  \end{aligned}
\end{align}
Since we are considering local interaction, the
evolution commute with the total translation
operator            
\begin{align}
  \label{eq:9}
\mathbf{T}\otimes \mathbf{T} \otimes I   \ket{x_1}\ket{x_2}  \ket{i}\ket{j}
 =
  \ket{x_1+1}\ket{x_2+1}  \ket{i}\ket{j}.
\end{align}
Therefore it 
convenient to introduce the relative
coordinate basis \footnote{we notice that only the pairs $(y,z)$ with $y$ and
$z$ both even or both odd correspond to integers
$x_1$ and $x_2$. However it is convenient to
consider the extended Hilbert space in which $y$
and $z$ run free.}
\begin{align}
  \label{eq:10}
  &
  \ket{x_1}\ket{x_2} \ket{i}\ket{j}\mapsto 
  \ket{z}\ket{y} \ket{i}\ket{j} + \\
  &y \coloneqq x_1-x_2 \quad z \coloneqq x_1+x_2 
\end{align}
and to Fourier transform the $z$ coordinate, i.e.
\begin{align}
  \label{eq:17}
   &
  \ket{p}\ket{y} \ket{i}\ket{j} \coloneqq  
  \left ( \sum_{z \in \mathbb{Z}} e^{-i p z}
     \ket{z} \right) \ket{y} \ket{i}\ket{j}.
\end{align}
In this basis, the free evolution $\mathbf{F}
\otimes\mathbf{F}$ becomes
\begin{align}
  \label{eq:18}
    \begin{aligned}
      \mathbf{F} \otimes \mathbf{F} =&
      \int_{\pi}^\pi \ketbra{p}{p} \otimes
      \mathbf{F}_2(p) \\
      \mathbf{F}_2(p):=&
  e^{-ip} \mathbf{I} \otimes
\sum_{i,j \in \{1,4\}}  \ketbra{i}{i}\otimes
  \ketbra{j}{j} +\\
  & +\mathbf{T} \otimes
\sum_{\substack{i \in \{1,4\}\\ j\in \{2,3\}}}  \ketbra{i}{i}\otimes
  \ketbra{j}{j} +\\
  &+\mathbf{T}^\dag\otimes
\sum_{\substack{i\in \{2,3\} \\ j \in \{1,4\}  }}  \ketbra{i}{i}\otimes
  \ketbra{j}{j} + \\
  &+ e^{ip} \mathbf{I} \otimes
\sum_{\substack{i\in \{2,3\} \\ j \in \{2,3\}  }}  \ketbra{i}{i}\otimes
  \ketbra{j}{j} .    
  \end{aligned}
\end{align}
Let us now consider the interaction terms of
Equation \eqref{eq:3}. Since we are considering
the two particle sector only the 4-fermion
interactions acts nontrivially.  We
notice that the free evolution of Equation
\eqref{eq:7} does not change the relative
direction of propagation of the modes. In order to
have a richer phenomenology we investigate the
interaction with the constant $\lambda_1$ in
Equation~\eqref{eq:class} which couples modes with 
different relative direction of propagation.

\begin{figure}[t]
    \includegraphics
   [width=\columnwidth]{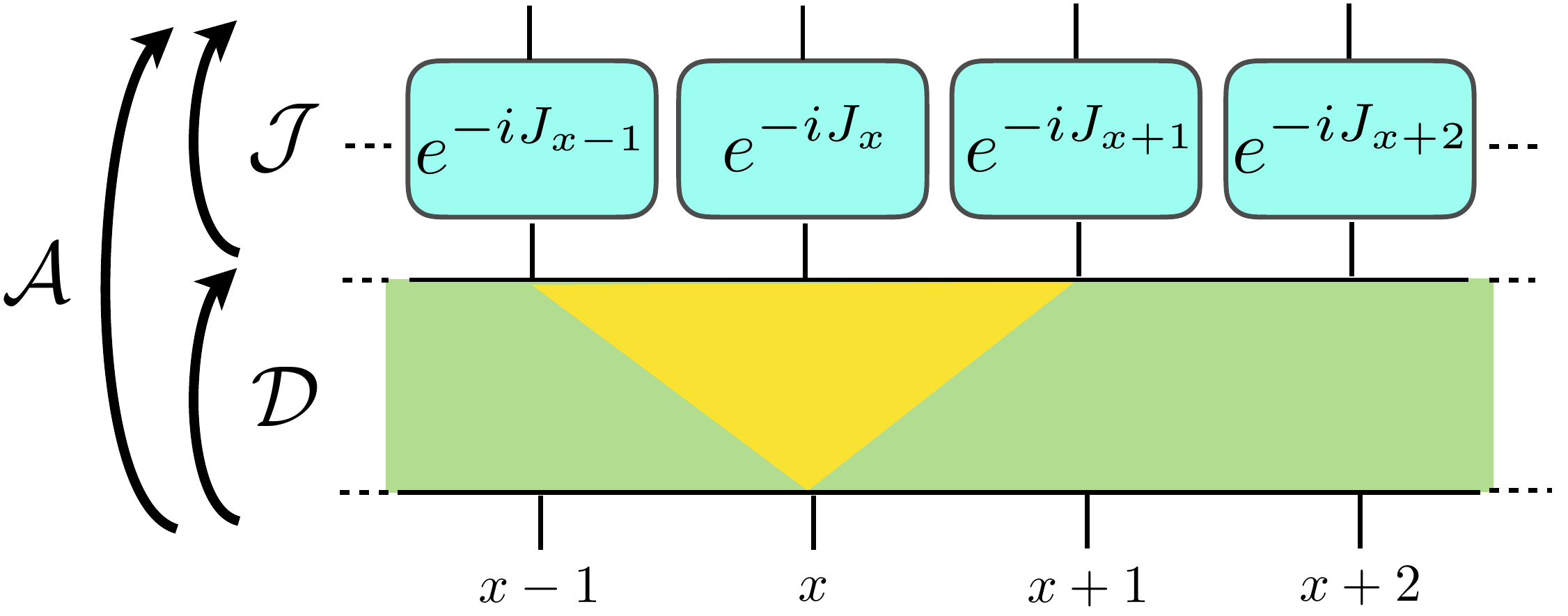}
   \caption{ (Colors online) 
The single step evolution of the FCA of
Equation~\eqref{eq:1}. Each site of the lattice
corresponds to a four-component fermionic
field. The free evolution $\mathcal{D}$ involves
the nearest neighbors while the  interaction $\mathcal{J}$ is completely
local. }
   \label{fig:interactingQCA}
 \end{figure}
Then, the interaction term reads as follows:
\begin{align}
  \label{eq:6}
 & \mathcal{J} \big( \psi_{x,a} \big) \coloneqq
\exp ( -i \sum_xJ_x)
  \psi_{x,a}\exp \big( i \sum_xJ_x\big), \\
  &\begin{aligned}
    \nonumber
  &J_x \coloneqq \lambda O_x +
  \overline{\lambda}O^\dagger_x , \qquad \lambda\in\mathbb{C},\\
  &O_x \coloneqq \psi_{x,2}
                   \psi^\dagger_{x,3}
                   \psi_{x,3}
                   \psi^\dagger_{x,4} +
                   \psi_{x,1}
                   \psi^\dagger_{x,3}
                   \psi^\dagger_{x,4}
                   \psi_{x,4} \,.
  \end{aligned}
\end{align}


From a straightforward computation it
follows that the two-particle sector of
$\mathcal{J}$ is described by the unitary operator
$ \mathbf{J}_2 \coloneqq \mathbf{P}_- \mathbf{J}
\mathbf{P}_-$, where
\begin{align}
  \label{eq:11}
  \begin{aligned}
  &\mathbf{J} \coloneqq \exp \big(-i \sum_{x}\mathbf{J}_{x} \big),
  \\
  &\mathbf{J}_{x} \coloneqq\big( \lambda \mathbf{O} +
  \overline{\lambda} \mathbf{O}^\dag \big) \otimes \ketbra{x}{x} \otimes \ketbra{x}{x}\\  
&\begin{aligned}
\mathbf{O}\coloneqq &\ketbra{3}{1} \otimes\ketbra{4}{4} +
\ketbra{3}{3}\otimes \ketbra{4}{2} +\\
&+\ketbra{4}{2} \otimes\ketbra{3}{3} +
\ketbra{4}{4}\otimes \ketbra{3}{1} \, .
\end{aligned}
\end{aligned}
  \end{align}
Since   $[\mathbf{F} \otimes \mathbf{F} ,
\mathbf{P}_- ] = [\mathbf{J}, \mathbf{P}_-] =0$,
  we have that the  two-particle sector of $\mathcal{A}$
is described by the operator
\begin{align}
  \label{eq:26}
  \mathbf{A}_2 \coloneqq \mathbf{P}_- \, \Big( \mathbf{J} \,
  (\mathbf{F} \otimes \mathbf{F}  )   \Big)\, \mathbf{P}_- .
\end{align}
Therefore, the goal of this section is to
diagonalize the operator $\mathbf{A}_2$.
Let us now define 
\begin{align}
  \label{eq:20}
  &\mathbf{Q} \coloneqq \mathbf{Q}' \otimes
      \mathbf{I} \otimes \mathbf{I}\\
    &\begin{aligned}
  \mathbf{Q}'\coloneqq
  &\ketbra{2}{2} \otimes \ketbra{3}{3} + \ketbra{3}{3} \otimes \ketbra{2}{2} +\\
  +&\ketbra{3}{3} \otimes \ketbra{4}{4} +
  \ketbra{4}{4} \otimes \ketbra{3}{3}
+\\
    +&\ketbra{1}{1} \otimes \ketbra{4}{4} 
  +\ketbra{4}{4} \otimes \ketbra{1}{1} .  
  \end{aligned}
\end{align}
From a straightworward computation we have
$[\mathbf{A}_2, \mathbf{Q}] = 0$ an therefore we
can write
\begin{align}
  \label{eq:21}
  \begin{aligned}
&\mathbf{A}_2 \coloneqq \widetilde{\mathbf{A}}_2 \oplus
                  \widetilde{\mathbf{A}}^\perp_2
               \,  \\
  &\widetilde{\mathbf{A}}_2, \coloneqq
  \mathbf{Q}\mathbf{A}_2 \mathbf{Q}, \qquad
  \widetilde{\mathbf{A}}^\perp_2 \coloneqq
  (\mathbf{I} - \mathbf{Q}) \mathbf{A}_2 (\mathbf{I} -
  \mathbf{Q}).     
  \end{aligned}
\end{align}
                              Since
                              $\mathbf{Q}\mathbf{J}
                              = \mathbf{J}
                                \mathbf{Q} = 0 $,
                                we have that 
                              $\widetilde{\mathbf{A}}^\perp_2$
is a free evolution. Therefore, we 
                              focus 
                              our analysis on the
                              non trivial term
                              $\widetilde{\mathbf{A}}_2$.
Let us define the vectors
\begin{align}
  \label{eq:13}
  \begin{aligned}
  \ket{e_1} \coloneqq \ket{1}\ket{4} 
\quad 
 \ket{e_2} \coloneqq \ket{3}\ket{2} \quad
 \ket{e_3} \coloneqq \ket{3}\ket{4}\\
 \ket{e_4} \coloneqq \ket{4}\ket{1}\quad
\ket{e_5} \coloneqq \ket{2}\ket{3} 
\quad \ket{e_6} \coloneqq \ket{4}\ket{3}  
  \end{aligned}
\end{align}
and the Hilbert space
$\widetilde{\mathcal{H}}\coloneqq \spn \{ \ket{e_i}, \}$. In the basis of Eq.~\eqref{eq:13}, we have:
\begin{align}
  \label{eq:16}
\lambda  \mathbf{O} + \overline{\lambda}
  \mathbf{O}^\dagger  =
\begin{pmatrix}
  {\mathbf{O}}'  & \mathbf{0}  \\
\mathbf{0}  &  {\mathbf{O}}'
\end{pmatrix},\qquad
   {\mathbf{O}}' \coloneqq
  \begin{pmatrix}
0 & 0 & -\bar{\lambda}   \\
0 & 0 & -\bar{\lambda}  \\
-\lambda & -\lambda & 0  \\
\end{pmatrix}.
\end{align}
 If we consider the total translation operator                 
\begin{align}
  \label{eq:9}
  \mathbf{T}_2  \ket{e_i}
  \ket{x_1}\ket{x_2}  = \ket{e_i}
  \ket{x_1+1}\ket{x_2+1}  
\end{align}
it is  straightforward to verify the commutation relation
$ [\mathbf{T}_2,  \widetilde{\mathbf{A}}_2] =0 $.
Therefore, it is convenient to introduce the relative
coordinate basis \footnote{we notice that only the pairs $(y,z)$ with $y$ and
$z$ both even or both odd correspond to integers
$x_1$ and $x_2$. However it is convenient to
consider the extended Hilbert space in which $y$
and $z$ run free.}
\begin{align}
  \label{eq:10}
  &\ket{e_i}
  \ket{x_1}\ket{x_2} \mapsto \ket{e_i}
  \ket{y}\ket{z}  \\
  &y \coloneqq x_1-x_2 \quad z \coloneqq x_1+x_2 
\end{align}
and to Fourier transform the $z$ coordinate, i.e.
\begin{align}
  \label{eq:17}
   &\ket{e_i}
  \ket{y}\ket{p} \coloneqq  \ket{e_i}
  \ket{y} \sum_{z \in \mathbb{Z}} e^{-i p z} \ket{z}.
\end{align}
In this basis, the operator
$\widetilde{\mathbf{A}}_2$ can be written as
follows:
\begin{align}
  \label{eq:19}
  &\widetilde{\mathbf{A}}_2 = \int_{-\pi}^{\pi}\!\! \! dp \, \,\,
                              \widetilde{\mathbf{A}}_2
                             (p) \otimes \ketbra{p}{p}, \\
  &\widetilde{\mathbf{A}}_2(p) \coloneqq 
                              \widetilde{\mathbf{P}}_- \,
                                \widetilde{\mathbf{J}} \,
                             \widetilde{\mathbf{D}}_p
                                \,
                                \widetilde{\mathbf{P}}_-, \label{eq:23}\\
  \label{eq:8}
  &\widetilde{\mathbf{J}} \coloneqq \exp
                           \left( -i \begin{pmatrix}
  {\mathbf{O}}'  & \mathbf{0}  \\
\mathbf{0}  &  {\mathbf{O}}'
\end{pmatrix}
                           \otimes \ketbra{0}{0} \right)
                            , \\
   &\widetilde{\mathbf{D}}_p\coloneqq \begin{pmatrix}
e^{2ip} & 0 & 0 & 0 & 0 & 0\\
0 & e^{-2ip} & 0 & 0 & 0 & 0\\
0 & 0 & \mathbf{T}_y^{\dagger 2} & 0 & 0 & 0 \\
0 & 0 & 0 & e^{2ip} & 0 & 0 \\
0 & 0 & 0 & 0 & e^{-2ip} & 0 \\
0 & 0 & 0 & 0 & 0 & \mathbf{T}_y^{2}  \\
\end{pmatrix},\\
  &\widetilde{\mathbf{P}}_-\coloneqq \frac12
                            \left(\mathbf{I}\otimes
                            \mathbf{I} -
                            \begin{pmatrix}
{\mathbf{0}}  & \mathbf{I}  \\
\mathbf{I}  &  {\mathbf{0}} 
\end{pmatrix}\otimes \sum_y\ketbra{-y}{y} \right).                
\end{align}
Since the eigenvalues of $ {\mathbf{O}}'$ are
$0, \pm\sqrt{2}|\lambda|$ we have that
$\widetilde{\mathbf{J}}$ is the identity
if $\sqrt{2}|\lambda|$ is an integer multiple of $2\pi$. 
Therefore, we will assume that
$\sqrt{2}|\lambda| \neq 2n 
\pi$ ($n \in \mathbb{Z}$).
Then, we can focus on the
diagonalization of the (infinite-dimensional)
operator $\widetilde{\mathbf{A}}_2(p)$.  Let us
now consider the following subspaces (see also Fig. \ref{fig:subspaces}):
\begin{align}
  \label{eq:22}
  &\begin{aligned}
    &\widetilde{\mathcal{H}}_a  \coloneqq \spn \{ \ket{e_1}
  \ket{y} - \ket{e_4}  \ket{-y} , y \neq 0 \} ,  \\
   &\widetilde{\mathcal{H}}_b  \coloneqq \spn \{ \ket{e_2}
  \ket{y} - \ket{e_5}\ket{-y} , y \neq 0 \} , \\
 &    \widetilde{\mathcal{H}}_c  \coloneqq \spn \{ \ket{e_3}
  \ket{2y+1} - \ket{e_6}
  \ket{-2y-1} \}, \\
 & \widetilde{\mathcal{H}}_d  \coloneqq
  \big(
  \widetilde{\mathcal{H}}_a \oplus
  \widetilde{\mathcal{H}}_b \oplus
  \widetilde{\mathcal{H}}_c
  \big)^\perp.
\end{aligned}
\end{align}
We remark that
\begin{align}
\supp \widetilde{\mathbf{P}}_-= & \widetilde{\mathcal{H}}_a \oplus
  \widetilde{\mathcal{H}}_b \oplus
  \widetilde{\mathcal{H}}_c \oplus
  \widetilde{\mathcal{H}}_d.
\end{align}
\begin{figure}[t]
      \includegraphics
   [width=0.8\columnwidth]{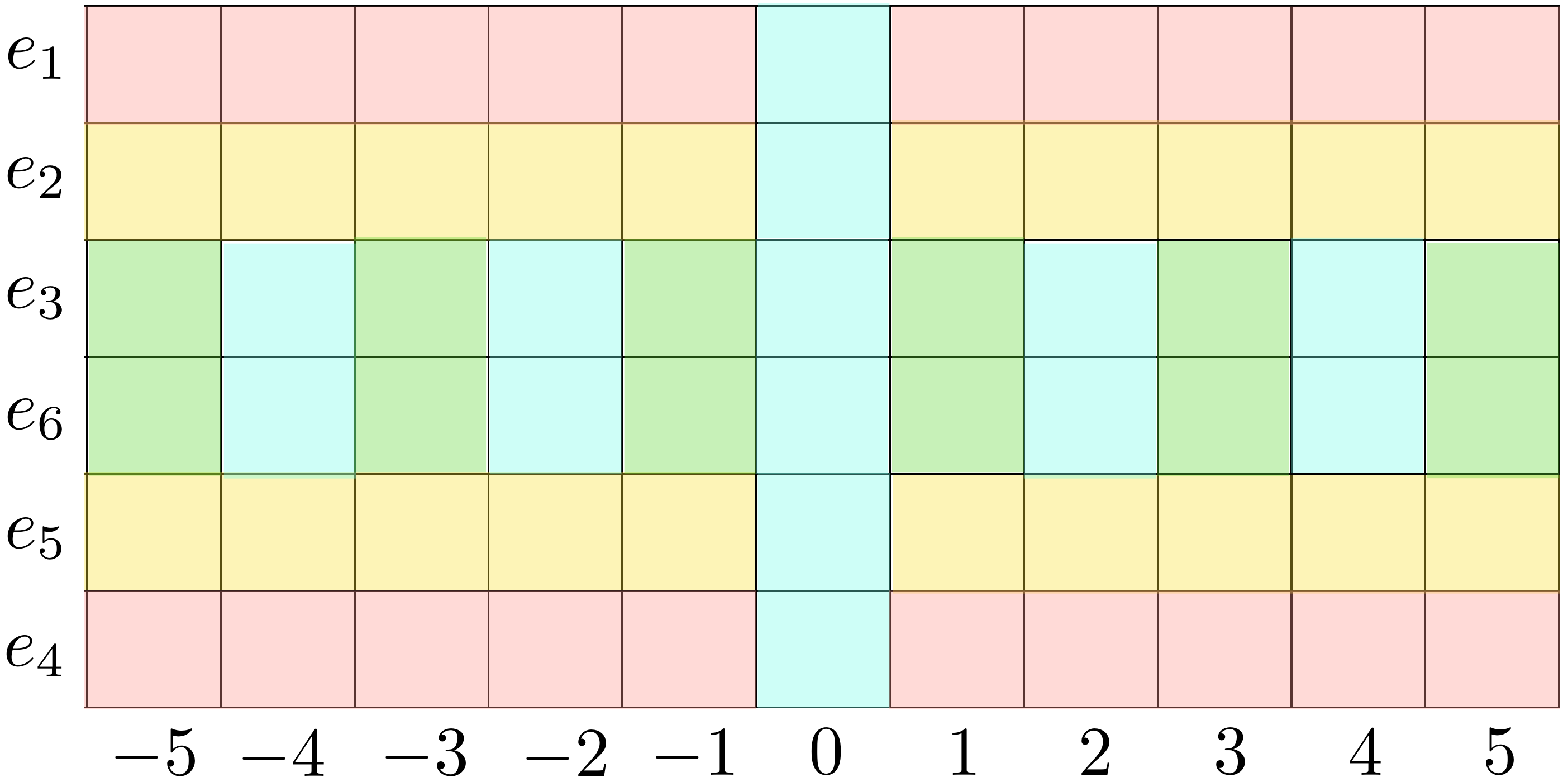}
   \caption{(Colors online) Schematic
     representation of the decomposition of
     Equation \eqref{eq:22}.  On the horizontal
     axis we have the spatial degrees of freedom,
     on the vertical axis the internal ones.
     $\widetilde{\mathcal{H}}_a$ corresponds to
     the red region, $\widetilde{\mathcal{H}}_b$
     corresponds to the yellow region,
     $\widetilde{\mathcal{H}}_c$ corresponds to
     the green region and
     $\widetilde{\mathcal{H}}_d$ corresponds to
     the blue region.  Only states whose support
     is in $\widetilde{\mathcal{H}}_d$ are
     affected by the interaction. }
  \label{fig:subspaces}
\end{figure}
We are now ready to proceed with the derivation of
the spectral resolution of
$\widetilde{\mathbf{A}}_2(p)$.
For the sake of clarity, we split the proof
into several lemmas. 
\begin{lemma}
  \label{lmm:specrestrivialpart}
  Let $\widetilde{\mathbf{A}}_2(p)$ be defined as
  in Equation \eqref{eq:23}. Then we have
  \begin{align}
    \label{eq:24}
    \begin{aligned}
    \widetilde{\mathbf{A}}_2(p) &= e^{i2p}
    \widetilde{\mathbf{P}}_a + e^{-i2p}
    \widetilde{\mathbf{P}}_b  + \\
    & +\widetilde{\mathbf{P}}_c  
    \widetilde{\mathbf{A}}_2(p)
    \widetilde{\mathbf{P}}_c  +
    \widetilde{\mathbf{P}}_d
    \widetilde{\mathbf{A}}_2(p) \widetilde{\mathbf{P}}_d  
    \end{aligned}
  \end{align}
  where 
  $\widetilde{\mathbf{P}}_i$, denotes the orthogonal
    projector on $\widetilde{\mathcal{H}}_i$.
\end{lemma}
\begin{proof}
  The subspaces $\widetilde{\mathcal{H}}_i$
  in Equation  \eqref{eq:22} are invariant
  subspaces of $\widetilde{\mathbf{A}}_2(p)$.
  One can directly verify that
  $\widetilde{\mathbf{P}}_a\widetilde{\mathbf{A}}_2(p)
  \widetilde{\mathbf{P}}_a =
  e^{i2p}\widetilde{\mathbf{P}}_a$
  and
  $\widetilde{\mathbf{P}}_b \widetilde{\mathbf{A}}_2(p)
  \widetilde{\mathbf{P}}_b =
  e^{-i2p}\widetilde{\mathbf{P}}_b$.
\end{proof}
\begin{lemma}
\label{lmm:specresfreepart}
  The operator $\widetilde{\mathbf{P}}_c  
    \widetilde{\mathbf{A}}_2(p)
    \widetilde{\mathbf{P}}_c $ has the following
    spectral resolution:
    \begin{align}
      \label{eq:25}
      \begin{aligned}
             \widetilde{\mathbf{P}}_c  
    \widetilde{\mathbf{A}}_2(p)
    \widetilde{\mathbf{P}}_c  &= \int_{-\pi} ^\pi
    \!\! \!\! dk \, 
      e^{-ik}\ketbra{\phi_f(k)}{\phi_f(k)},\\
      \ket{\phi_f(k)} &\coloneqq \frac{1}{\sqrt2} \big(
      \ket{e_3}\ket{k}_o - \ket{e_6}\ket{-k}_o\big) \\
      \ket{k}_o& \coloneqq \sum_{y\in\mathbb{Z} }
      \frac{1}{\sqrt{2 \pi}} e^{-iky} \ket{2y+1}.
      \end{aligned}
    \end{align}
\end{lemma}
\begin{proof}
  We have
  $ \widetilde{\mathbf{P}}_c
  \widetilde{\mathbf{A}}_2(p)
  \widetilde{\mathbf{P}}_c =
  \widetilde{\mathbf{P}}_c
\,
  \widetilde{\mathbf{J}} \,
  \widetilde{\mathbf{D}}_p \,
  \widetilde{\mathbf{P}}_c $
Let us define the subspace
  $\ell^2(2\mathbb{Z}+1) \coloneqq
  \spn{\{\ket{2y+1}, y\in \mathbb{Z}\}} \subset
  \ell^2(\mathbb{Z})$ and let us
  denote with
  $\widetilde{\mathbf{P}}_\mathrm{odd} $ the projector on
  $\ell^2(2\mathbb{Z}+1)$.
  It is easy to verify that
   $\widetilde{\mathbf{P}}_c \big(
   \mathbf{I}
  \otimes \widetilde{\mathbf{P}}_\mathrm{odd}\big) =
   \big(
   \mathbf{I}
  \otimes \widetilde{\mathbf{P}}_\mathrm{odd}\big)\widetilde{\mathbf{P}}_c
  =\widetilde{\mathbf{P}}_c$.
  Since $\widetilde{\mathbf{J}}$ is localized at
  the origin, we have that
  \begin{align*}
 \widetilde{\mathbf{P}}_c
\,
  \widetilde{\mathbf{J}} \,
  \widetilde{\mathbf{D}}_p \,
  \widetilde{\mathbf{P}}_c=&
  \widetilde{\mathbf{P}}_c
  \, \big(
   \mathbf{I}
  \otimes \widetilde{\mathbf{P}}_\mathrm{odd}\big)
\,
  \widetilde{\mathbf{J}} \,
  \widetilde{\mathbf{D}}_p \,
  \big(
   \mathbf{I}
  \otimes \widetilde{\mathbf{P}}_\mathrm{odd}\big)\,
  \widetilde{\mathbf{P}}_c =\\
=&\widetilde{\mathbf{P}}_c
  \, \big(
   \mathbf{I}
  \otimes \widetilde{\mathbf{P}}_\mathrm{odd}\big)
\,
  \widetilde{\mathbf{D}}_p \,
  \big(
   \mathbf{I}
  \otimes \widetilde{\mathbf{P}}_\mathrm{odd}\big)\,
  \widetilde{\mathbf{P}}_c =\\
=&\widetilde{\mathbf{P}}_c \,\big(
  \ketbra{e_3}{e_3}\otimes
    \mathbf{P}_\mathrm{odd}\mathbf{T}^{\dag 2}_y
    \mathbf{P}_\mathrm{odd} \big)
    \widetilde{\mathbf{P}}_c+ \\
   &+ \widetilde{\mathbf{P}}_c \,\big(
  \ketbra{e_6}{e_6} \otimes
    \mathbf{P}_\mathrm{odd}\mathbf{T}^{2}_y
    \mathbf{P}_\mathrm{odd} \big) \widetilde{\mathbf{P}}_c
  \end{align*}
 The operator
$  \mathbf{P}_\mathrm{odd}\mathbf{T}^{2}_y
    \mathbf{P}_\mathrm{odd} $ is easily diagonalized as
    follows:
    \begin{align*}
      \mathbf{P}_\mathrm{odd}\mathbf{T}^{2}_y
    \mathbf{P}_\mathrm{odd}  = \int_{-\pi}^\pi \!\!\! dk \, e^{-ik}
      \ket{k}_o\bra{k}_o.
    \end{align*}
    Then we have
    \begin{align*}
  &\ketbra{e_3}{e_3}\otimes
    \mathbf{P}_\mathrm{odd}\mathbf{T}^{\dag 2}_y
    \mathbf{P}_\mathrm{odd} +
  \ketbra{e_6}{e_6} \otimes
    \mathbf{P}_\mathrm{odd}\mathbf{T}^{2}_y
      \mathbf{P}_\mathrm{odd}=\\
    &  =\int_{-\pi}^\pi\! \!\!\! dk \, e^{-ik}
      \big( \ketbra{e_3}{e_3}\otimes
      \ket{k}_o\bra{k}_o +
      \ketbra{e_6}{e_6}\otimes
      \ket{-k}_o\bra{-k}_o \big).
    \end{align*}
     Since $\widetilde{\mathbf{P}}_c
     \ket{e_3}\ket{k}_o =- \widetilde{\mathbf{P}}_c
     \ket{e_6}\ket{-k}_o = 1/\sqrt2 \ket{\phi_f(k)} $
     we have the thesis.
\end{proof}

\begin{lemma}
\label{lmm:specresscatbound}
  We have the following
  spectral resolutions:
  \begin{enumerate}
  \item $\sqrt2 |\lambda| \neq n\pi $ and $2p \neq n' \pi$ with $n,n' \in \mathbb{Z} $
      \begin{align}
        \label{eq:28}
&  \widetilde{\mathbf{P}}_d
  \widetilde{\mathbf{A}}_2(p)
  \widetilde{\mathbf{P}}_d =
\int_{-\pi} ^\pi
    \!\! \!\! dk \, 
      e^{-ik}\ketbra{\phi_s(k)}{\phi_s(k)},
  \end{align}
  
    \item
$\sqrt2 |\lambda| \neq n \pi $ and $2p = n' \pi$ with $n,n' \in \mathbb{Z} $
      \begin{align}
        \begin{aligned}
  \widetilde{\mathbf{P}}_d
  \widetilde{\mathbf{A}}_2(p)
  \widetilde{\mathbf{P}}_d = &
\int_{-\pi} ^\pi
    \!\! \!\! dk \, 
                             e^{-ik}\ketbra{\phi_s(k)}{\phi_s(k)} + \\
        & + (-1)^{n'} \ketbra{\phi_{b0}}{\phi_{b0}}
    \end{aligned} 
      \end{align}
    \item $\sqrt2 |\lambda| = (2n+1) \pi $ with $n \in \mathbb{Z}$
      \begin{align}
        \begin{aligned}
  \widetilde{\mathbf{P}}_d
  \widetilde{\mathbf{A}}_2(p)
  \widetilde{\mathbf{P}}_d = &
\int_{-\pi} ^\pi
    \!\! \!\! dk \, 
                             e^{-ik}\ketbra{\phi_s(k)}{\phi_s(k)} + \\
                             & + \ketbra{\phi_{b-}}{\phi_{b-}} -
                             \ketbra{\phi_{b+}}{\phi_{b+}}
    \end{aligned} 
      \end{align}
      
  \end{enumerate}
where we defined
  \begin{align}
                              &
                                \label{eq:15}
      \begin{aligned}
        N_k \ket{\phi_s(k)} &\coloneqq \sum_{y\geq1}
        \overline{c_k}e^{-iyk}
        \Big(\ket{e_3}\ket{2y}-\ket{e_6}\ket{-2y}\Big
        )+
        \\
        & + \sum_{y\leq 0}
        {c_k}e^{-iyk}\Big(\ket{e_3}\ket{2y}-\ket{e_6}\ket{-2y}\Big
        )
        \\
        & +\ket{\psi_0} \ket{0},
    \end{aligned}
 \\
   &
     \begin{aligned}
 \ket{\psi_0}\coloneqq &d_k(\ket{e_1}-\ket{e_4})-\overline{d_{-k}}(\ket{e_2}-\ket{e_5})  
\end{aligned}\\
    &\begin{aligned}
 \ket{\phi_{b0}}\coloneqq &\frac12 \Big(\ket{e_1}-\ket{e_4}-\ket{e_2}+\ket{e_5} \Big) \ket{0}  
\end{aligned}\\
         &\begin{aligned}
           \ket{\phi_{b-}}\coloneqq &\frac12 \Big(e^{-ip}(\ket{e_1}-\ket{e_4})+ \\
           & - e^{ip}(\ket{e_2}-\ket{e_5}) \Big) \ket{0}  
         \end{aligned}\\
                              &\begin{aligned}
                                \ket{\phi_{b+}}\coloneqq &\frac12 \Big(e^{-ip}(\ket{e_1}-\ket{e_4})+\\
                                &+
                                e^{ip}(\ket{e_2}-\ket{e_5})
                                \Big) \ket{0}
     \end{aligned}\\
&
\begin{aligned}
   c_k\coloneqq &  \cos(\sqrt{2}|\lambda|)(e^{-ik}-
         \cos(2p)) +\\
&+ (e^{ik}- \cos(2p)), 
\end{aligned}\\
    &d_k \coloneqq \frac{i \overline{\lambda}
      \sin(\sqrt{2}|\lambda|)}{\sqrt2 |\lambda|}
      (e^{-ik}-e^{-i2p}).
  \end{align}
and $N_k$ is a suitable normalization constant
such that $\braket{\phi_s(k)}{\phi_s(k')} =
\delta(k-k').$
\end{lemma}
\begin{proof}
  Let us consider the first case
  One can verify by direct computation that
  $\overline{N_k}N_{k'}\braket{\phi_s(k)}{\phi_s(k')}
  \propto \delta(k-k') $ and that
  $ N_k\ket{\phi_s(k)} $ are improper eigenstates of
 $  \widetilde{\mathbf{P}}_d
  \widetilde{\mathbf{A}}_2(p)
  \widetilde{\mathbf{P}}_d $
  with eigenvalue
  $e^{-ik} $. Similarly,
  one also verifies (in the appropriate cases) that
$\ket{\phi_{b0}}$ ,   
$\ket{\phi_{b+}}$ and
$\ket{\phi_{b-}}$ are also eigenstates.
  One then have to verify the completeness
  relation
  $ \int_{-\pi}^{\pi}\ketbra{\phi_s(k)}{\phi_s(k)}
  = \widetilde{\mathbf{P}}_d$ (and the
  compleneteness of the other spectral
  resolutions).  This will be done in Appendix
  \ref{sec:compl-solut}
\end{proof}

We can combine the statements of the Lemmas
\ref{lmm:specrestrivialpart},~\ref{lmm:specresfreepart}
and \ref{lmm:specresscatbound}  into a single
proposition which summarizes the results of this section.

\begin{proposition}\label{prp:propositionfinale}
  The unitary operator $\mathbf{A}_2$, defined in
  Equation \eqref{eq:26}, describes the single step
  evolution of the the two particle sector of the
  FCA $\mathcal{A}$ of Equation~\eqref{eq:1}.
  From Equation \eqref{eq:21} we have the
  following decomposition: 
\begin{align}
 \mathbf{A}_2 = \widetilde{\mathbf{A}}_2 \oplus
                  \widetilde{\mathbf{A}}^\perp_2.
\end{align}
The operator $\widetilde{\mathbf{A}}^\perp_2$ acts as the
free evolution and the operator $\mathbf{A}_2$ has the following
spectral resolution:
\begin{align}
  \label{eq:29}
&  \widetilde{\mathbf{A}}_2  = 
\int_{-\pi}^{\pi}\!\! \! dp \, \,\,
                              \widetilde{\mathbf{A}}_2
                             (p) \otimes \ketbra{p}{p},\\
&  \begin{aligned}
    \widetilde{\mathbf{A}}_2(p) =& e^{i2p}
    \widetilde{\mathbf{P}}_a + e^{-i2p}
    \widetilde{\mathbf{P}}_b  +
    \widetilde{\mathbf{B}} + \\
+\int_{-\pi} ^\pi &
    \!\! \!\! dk \, 
    e^{-ik} \Big( \ketbra{\phi_f(k)}{\phi_f(k)} +\ketbra{\phi_s(k)}{\phi_s(k)} \Big)  
  \end{aligned}\\
\nonumber
 & \begin{aligned}
    \widetilde{\mathbf{B}} \coloneqq
    \left \{
      \begin{aligned}
         &     (-1)^{n'} \ketbra{\phi_{b0}}{\phi_{b0}} 
      &\sqrt2 |\lambda| \neq (n+1) \pi \\ &&\mbox{and }2p = n' \pi,   
      \\
&\ketbra{\phi_{b-}}{\phi_{b-}} -
\ketbra{\phi_{b+}}{\phi_{b+}} &     \sqrt2 |\lambda| = (n+1) \pi,   \\
&0 &\mbox{otherwise.}
      \end{aligned}
    \right .
  \end{aligned}
\end{align}
where we remind that $\ket{p}$ are the (improper)
eigenstates of the total momentum and that
$\widetilde{\mathbf{A}}_2(p) $ acts on the
relative coordinate and the internal degrees of freedom.
\end{proposition}

\begin{figure}
  \centering
  \includegraphics[width=0.5\textwidth]{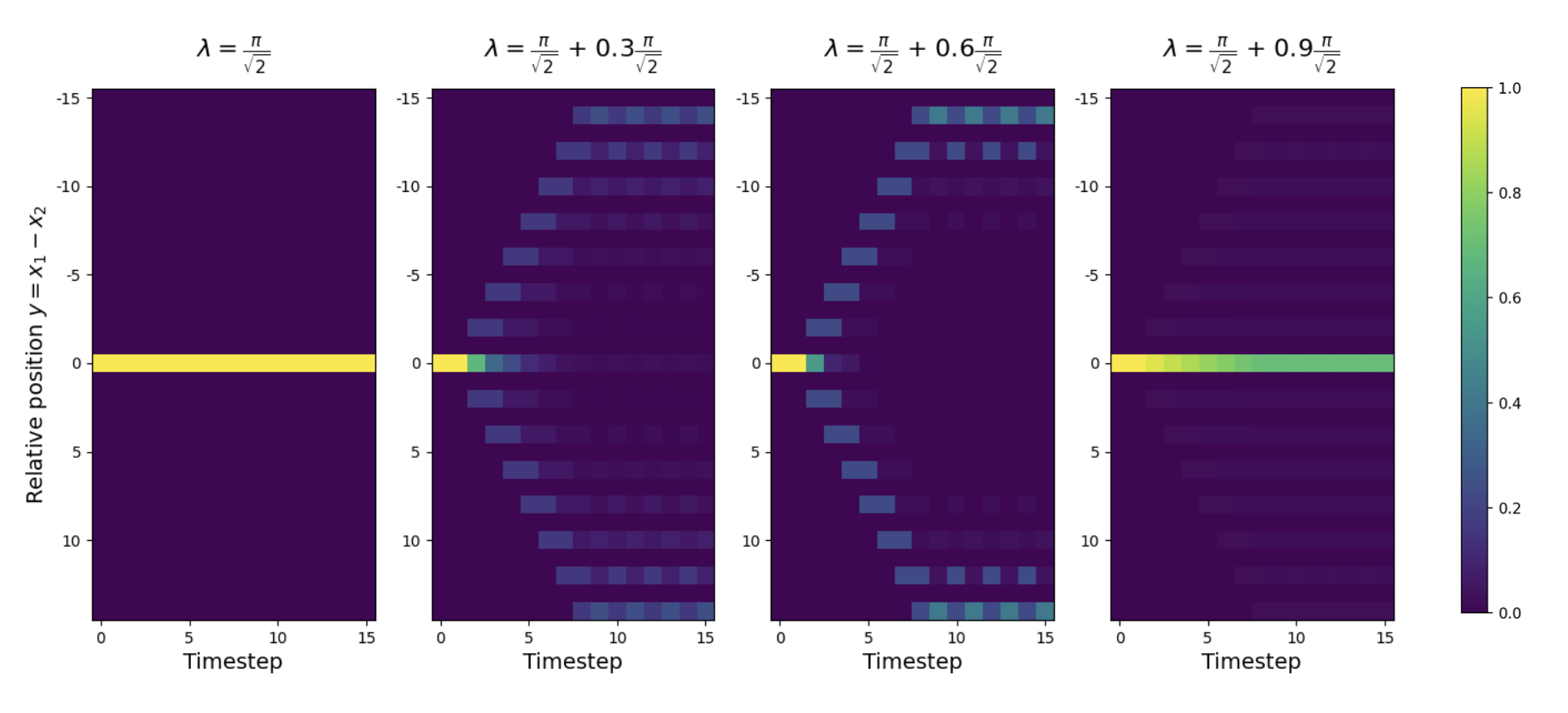}
  \caption{(Colors online) Numerical simulation of
    $\widetilde{\mathbf{A}}_2(p)$
    with
    $p={\pi}/{4}$ and $\lambda \in
    ({\pi}/{\sqrt{2}},{2\pi}/{\sqrt{2}})$.
    On the $x$ axis we report the simulation
    time-step, while on the $y$ axis the relative position
    of the two particles. The probability of
    finding the particles at a certain relative
    position after measuring over the
    internal degrees of freedom is reported in
    color scale.  Bound states are
    observed around values
    $\lambda=n{\pi}/{\sqrt{2}}$,
    $n\in \mathbb{Z}$, while scattering states
    arise for values in between two subsequent
    integer multiples of ${\pi}/{\sqrt{2}}$.
    The initial state is
    $|\phi_{b+}\rangle = \frac{1}{2}\left( e^{-ip}
      \left( |e_1 \rangle - |e_4 \rangle \right) +
      e^{+ip} \left( |e_2 \rangle - |e_5 \rangle
      \right) \right)|0\rangle$.}
  \label{fig:QCA_sim2d}
\end{figure}

\section{Conclusions}
\label{sec:4}
In this work we have studied the two particle
sector of a interacting FCA. The free evolution
is modeled after the massless Dirac quantum walk
\cite{PhysRevA.90.062106} while the interaction is
quartic, local and number preserving. The
interaction has been chosen among the ones
with these properties which are
invariant under the isotropy group of the Weyl
cellular automata.
Those interactions  has been classified
with the help of our in-house code employing the
symbolic mathematics module \texttt{SymPy}.
The code is available on GitHub and the result of
the classification is given in Equation \eqref{eq:class}.

For the specific case considered, the dynamics of the two particle sector,
for a fixed value of the total momentum, can be described as a
single massless particle (the relative position) moving on a one
dimensional lattice with a potential which is
localized at the origin and that interacts with the
internal degrees of freedom. The spectral
resolution of the evolution operator provides a
complete account of the dynamics.  Referring to
Proposition \ref{prp:propositionfinale} we have
that: $i)$ the projectors
$\widetilde{\mathbf{P}}_a$ and
$\widetilde{\mathbf{P}}_b$ project on subspaces
which are orthogonal to the support of the
interaction term and on which the free evolution
acts only as total translation of the center of
mass. States in these subspaces describe particles
that experience no relative motion and no
interaction.  $ii)$ The (improper) eigenstates
$\ket{\phi_f(k)}$ describe free particles that
experience relative motion but do not feel the
presence of the interaction because they never
happen to be in the same site (in the relative
coordinate they stay confined on the odd sites
while the interaction is localized at $0$).
$iii)$ The (improper) eigenstates
$\ket{\phi_s(k)}$ describe particles that move
relatively to each other and that can
interact. These states describe the scattering
processes.  $iv)$ The states $\ket{\phi_{b0}}$ ,
$\ket{\phi_{b+}}$ and $\ket{\phi_{b-}}$ describe
bound states (see Fig.~\ref{fig:QCA_sim2d}). 
The latter are perfectly localized and
they occur only for some critical value of the
total momentum $p$ or some critical value of the
coupling constant $\lambda$.  The origin of this
phenomenon can be traced back to the time
discreteness of the QCA model.  When we
exponentiate an Hamiltonian operator $H$, the
resulting unitary operator $\exp(-iH)$ can have
degeneracies because of the periodicity of the
exponential function. This is what happens to the
free evolution for $2p = n\pi$ and to the
interacting term for
$\sqrt{2}|\lambda| = (2n+1) \pi$.  At those
critical values we find common eigestates of the
free evolution $\widetilde{\mathbf{D}}_p$ and of
the interaction $ \widetilde{\mathbf{J}}
$. Further studies are planned to determine
whether the 1-particle and 2-particle solutions
provide sufficient information to solve the
dynamics in all number sectors. This would open
the route for a statistical analysis of the model,
in order to determine whether it admits of phase
transitions or other critical phenomena. Another
interesting development consists in the study of
bound states as virtual particles and their
effective interactions with other stationary
states.

\section*{Code Availability}
The codes used in this article to derive the 
most general interaction in our setting and to 
simulate the interacting QW studied are publicly 
available on GitHub repository at \url{ https://github.com/edoardo100/interacting_dirac_QCA}. 

\acknowledgments
PP and AB acknowledge financial support
from European Union -- Next Generation EU through the MUR project Progetti di Ricerca d'Interesse Nazionale (PRIN) QCAPP No. 2022LCEA9Y.
\bibliographystyle{apsrev4-1}
\bibliography{bibliography}

\appendix
\section{Completeness of the solutions}
\label{sec:compl-solut}
Let us now prove the compleneteness of the solutions set found
in the previous section. 
A vector $\ket{\chi}$ in
$\widetilde{\mathcal{H}}_d$ can be written as
follows
\begin{align}
  \label{eq:30}
  \begin{aligned}
  \ket{\chi} &=
  \ket{\chi_0}\ket{0}+
  \sum_y \overline{c_{3,y}}(\ket{e_3}\ket{2y} - \ket{e_6}\ket{-2y} )\\
  \ket{\chi_0} &\coloneqq   \overline{{c}_{1,0}}(\ket{e_1} - \ket{e_4} ) +
  \overline{{c}_{2,0}}(\ket{e_2} - \ket{e_5} )    
  \end{aligned}
\end{align}
From Equation \eqref{eq:15} we have that
\begin{align}
  \label{eq:32}
  &\braket{\chi}{\phi_s(k)} =2 f(k) \\
  \nonumber
  &\begin{aligned}
    f(k) =&
    {c_{1,0}}d_k +
    {c_{2,0}}\overline{d_{-k}}+
    \sum_{y \leq 0 }   {c}_{3,y}c_ke^{-iyk}+\\
    &+  \sum_{y \geq 1 }   {c}_{3,y}\overline{c_k}e^{-iyk}  
    =\\
    =&\sum_{y \leq -1}
    ( \gamma c_{3,y-1} -\epsilon c_{3,y}+c_{3,y+1})
    e^{-iyk}+\\
    &+ \gamma c_{3,-1}
    -\epsilon c_{3,0}
    + \sigma e^{-i2p} c_{1,0} 
    + \sigma e^{i2p} c_{2,0}
    + \delta c_{3,1} \\
& 
+ (\gamma c_{3,0}
    - \sigma c_{1,0} 
    - \sigma c_{2,0}
    -\epsilon c_{3,1}
    + \delta c_{3,2}) e^{-ik} \\
&+  \sum_{y \geq 2}
 (  c_{3,y-1} -\epsilon c_{3,y}+ \gamma c_{3,y+1}) e^{-iyk}
\end{aligned} \\
  \nonumber
  &
    \begin{aligned}
    &\gamma \coloneqq \cos(\sqrt2 |\lambda|), \quad
  \epsilon \coloneqq \cos(2p) (\gamma + 1) \\
  &\sigma \coloneqq -\frac{i \overline{\lambda}
      \sin(\sqrt{2}|\lambda|)}{\sqrt2 |\lambda|}      
    \end{aligned}.
\end{align}
Therefore the condition
$\braket{\chi}{\phi_s(k)} =0$ is equivalent to
$f(k) =0 $ for any $k$. Since $f(k)$ is a Fourier
series, the condition $f(k) =0 $ for any $k$
implies the following system of equations:
\begin{align}
  \label{eq:34}
  \begin{aligned}
    &0 = \gamma c_{3,y-1} -\epsilon c_{3,y}+c_{3,y+1} 
    \quad\qquad \mbox{for }  y \leq -1\\
    &\begin{aligned}
0=&\gamma c_{3,-1}
    -\epsilon c_{3,0}
    +\sigma e^{-i2p} c_{1,0} +\\
    &+ \sigma e^{i2p} c_{2,0} 
    + \delta c_{3,1} 
    \end{aligned}\\
   &0=  \gamma c_{3,0}
    -\sigma c_{1,0} 
    - \sigma c_{2,0}
    -\epsilon c_{3,1}
    + \delta c_{3,2} \\
 &0= c_{3,y-1} -\epsilon c_{3,y}+ \gamma c_{3,y+1} \quad\qquad \mbox{for } y \geq 2 . 
  \end{aligned}
\end{align}
Equation \eqref{eq:34} can be conveniently
rewritten as follows:
\begin{align}
  \label{eq:33}
    &0 = \gamma x_{n-1} -\epsilon x_{n}+x_{n+1} 
    \quad\qquad \mbox{for }  n \leq -1\\
    &\begin{aligned}\label{eq:27}
0=&\gamma x_{-1}
    -\epsilon x_{0}
    +\sigma e^{-i2p} x_{1} +\\
    &+ \sigma e^{i2p} x_{2} 
    + \gamma x_{3} 
    \end{aligned}\\
   &0=  \gamma x_{0}\label{eq:35}
    -\sigma x_{1} 
    - \sigma x_{2}
    -\epsilon x_{3}
    + \gamma x_{4} \\
 &0= x_{n-1} -\epsilon x_{n}+ \gamma x_{n+1} \quad\qquad \mbox{for } n \geq 4 . \label{eq:36}
\\
  &  x_n=
  \begin{cases}
    c_{3,n} & n \leq 0\\
    c_{1,0} & n =1 \\
    c_{2,0} & n=2 \\
    c_{3,n-2}&n\geq 3
   \end{cases}
\end{align}

If $2 \sqrt{2} |\lambda| =  n \pi$ then we have $\gamma = 0$ and  Equations \eqref{eq:33} to \eqref{eq:36}
simplifies as follows:
\begin{align}
\label{eq:37}
  &0 = -\epsilon x_{n}+x_{n+1} 
    \quad\qquad \mbox{for }  n \leq -1\\
  &\begin{aligned}
    \label{eq:39}
0=&
    -\epsilon x_{0}
    +\sigma e^{-i2p} x_{1} + \sigma e^{i2p} x_{2} 
    \end{aligned}\\
\label{eq:40}
  &0=  
    -\sigma x_{1} 
    - \sigma x_{2}
    -\epsilon x_{3}
 \\
\label{eq:38}
  &0= x_{n-1} -\epsilon x_{n} \quad\qquad \mbox{for } n \geq 4 . 
\end{align}
If also $\epsilon = 0$ then Equations
\eqref{eq:37} and \eqref{eq:38} implies
$x_n = 0 $ for $n \geq 3$ and $n \leq 0$.
On the other hand, if
$\epsilon \neq 0$  the linear recurrence relations
\eqref{eq:37} and \eqref{eq:38} have the solutions
\begin{align}
  \begin{aligned}
  x_{-n} = \epsilon^{-n}x_{0} \quad \mbox{for } n \geq 1,\\
  x_{n+3} = \epsilon^{-n}x_{3} \quad \mbox{for } n \geq 1.   
  \end{aligned}
 \end{align}
However, $\gamma = 0$ implies that
$|\epsilon|= |\cos(2p)| \leq 1$ and the convergence of
$\sum_n |x_n|^2$ implies that $x_{0} = x_3 =0$.
We then conclude
that $x_n = 0 $ for $n \geq 3$ and $n \leq 0$.

Let us now suppose that
$2 \sqrt{2} |\lambda| \neq n \pi$ which implies
$\gamma \neq 0$.  The linear recurrence relations
\eqref{eq:34} and \eqref{eq:36} can be rewritten as follows:
\begin{align}
  \label{eq:46}
   x_{-n-1} - \frac{\epsilon}{\gamma} x_{-n}
  + \frac{1}{\gamma}x_{-n+1} =0&
                                 \quad\qquad \mbox{for }  n \geq 1\\
  \label{eq:47}
  x_{n+1}  - \frac{\epsilon}{\gamma} x_{n}
  +\frac{1}{\gamma} x_{n-1}
   =0&
    \quad\qquad \mbox{for }  n \geq 4
\end{align}
If $\epsilon = 0$ we have
\begin{align}
  \label{eq:49}
   x_{-2n} =  \left(-\frac{1}{\gamma} \right )^n x_0&
                                 \quad\qquad \mbox{for }  n \geq 1,\\
  \label{eq:47}
  x_{2n+1} 
   = \left(-\frac{1}{\gamma} \right )^n x_3 &
    \quad\qquad \mbox{for }  n \geq 1.
\end{align}
Therefore, the convergence of $\sum_n |x_n|^2$
implies that $x_n = 0 $ for $n \geq 3$ and
$n \leq 0$.  Let us now suppose that $\gamma,\epsilon \neq
0$. The characteristic equations of the recurrence
relations \eqref{eq:46} and \eqref{eq:47} are the
same and have the roots
\begin{align}
  \label{eq:48}
  r_{\pm} = \frac{\epsilon \pm
  \sqrt{\epsilon^2 - 4 \gamma}}{2 \gamma}
\end{align}
If $\epsilon^2 - 4 \gamma = 0$ then $r_+ = r_- = r =
\frac{\epsilon}{2\gamma}$ and
\begin{align}
&
 x_{-n} = (a +b n)r^n                                \quad\qquad \mbox{for }  n \geq 0,\\
  &
x_{n+2} = (a' +b' n)r^n    \quad\qquad \mbox{for }  n \geq 1
\end{align}
for some coefficient $a$ , $a'$, $b$ and $b'$.
Since
$|r| = \frac{|\epsilon|}{2|\gamma|} =
\frac{2}{|\epsilon |} \geq 1$ the convergence of
of $\sum_n |x_n|^2$ implies that $x_n = 0 $ for
$n \geq 3$ and $n \leq 0$.  If
$\epsilon^2 \geq 4 \gamma \neq 0$ then $r_+ $ and
$ r_- $ are distinct and the solution of
the recurrence relation is
\begin{align}
&
 x_{-n} = ar_-^n +b r_+^n                                \quad\qquad \mbox{for }  n \geq 0,\\
  &
x_{n+2} =a'r_-^n +b' r_+^n        \quad\qquad \mbox{for }  n \geq 1.
\end{align}
for some for some coefficient $a$ , $a'$, $b$ and $b'$.
One can verify (see Appendix \ref{sec:study-r_pm-})
that $|r_\pm| \geq 1$ and
the convergence of
of $\sum_n |x_n|^2$ implies that $x_n = 0 $ for
$n \geq 3$ and $n \leq 0$.
We have then proved that Equation \eqref{eq:33}
and Equation \eqref{eq:36} implies
that $x_n = 0 $ for
$n \geq 3$ and $n \leq 0$.
Then, Equations \eqref{eq:39} and \eqref{eq:40} become
\begin{align}
  \label{eq:45}
  \begin{aligned}
0=&
    \sigma e^{-i2p} x_{1} + \sigma e^{i2p} x_{2} ,
  \\ 
  0=&  
    -\sigma x_{1} 
    - \sigma x_{2}.    
  \end{aligned}
\end{align}
If
$\sqrt2 |\lambda| \neq n \pi \iff \sigma \neq 0$
and $\sin(2p) \neq 0$ Equation \eqref{eq:45}
implies $x_1 = x_2 = 0$.  If
$\sigma \neq 0$ and $\sin(2p) = 0$ we have the
solution $x_1 = -x_2$ and the orthogonality with
$\ket{\phi_{b0}}$ gives
$x_1 = x_2 = 0$.
Similarly, if $\sigma = 0$ the orthogonality with
$\ket{\phi_{b1}}$ and $\ket{\phi_{b2}}$
implies $x_1 = x_2 = 0$.

\section{Study of $|r_{\pm}|$ }
\label{sec:study-r_pm-}
We will prove that $|r_{\pm}| > 1$ by considering
three cases.

\subsection{$\epsilon^2
  - 4\gamma >0$ and $\gamma >0$ }

If $\epsilon^2
  - 4\gamma >0$ then $r_{\pm} \in \mathbb{R}$. 
If $\gamma >0$ we have
\begin{align*}
  |r_{\pm}| = \frac{|\epsilon \pm \sqrt{\epsilon^2
  - 4\gamma}|}{2\gamma} \geq
   \frac{|\epsilon| - \sqrt{\epsilon^2
  - 4\gamma}}{2\gamma}. 
\end{align*}
Since $\epsilon^2 - 4\gamma > 0$ and $|\gamma|
\leq 1$ implies
$|\epsilon| -2|\gamma| >  0$ we have
\begin{align*}
&    \frac{|\epsilon| - \sqrt{\epsilon^2
  - 4\gamma}}{2\gamma} \geq 1 \iff
  |\epsilon| - 2\gamma \geq \sqrt{\epsilon^2
  - 4\gamma} \iff \\
  &(|\epsilon| - 2\gamma )^2 \geq \epsilon^2 \iff
    \gamma - |\epsilon| \geq -1 \iff \\
  &\gamma - |\cos(2p)(1+\gamma)| \geq -1 \iff\\
  &(1 - |\cos(2p)|)\gamma - |\cos(2p)| \geq -1
\end{align*}
which proves that $|r_\pm| \geq 1$ for $\gamma > 0$.

\subsection{$\epsilon^2
  - 4\gamma >0$ and $\gamma < 0$ }

If $\epsilon^2
  - 4\gamma >0$ then $r_{\pm} \in \mathbb{R}$. 
If $\gamma < 0$ we have
\begin{align*}
  |r_{\pm}| = \frac{|\epsilon \pm \sqrt{\epsilon^2
  - 4\gamma}|}{2\gamma} \geq
\frac{   \sqrt{\epsilon^2  + 4 |\gamma|} - |\epsilon|}{2|\gamma|}. 
\end{align*}
Then we have
\begin{align*}
 & \frac{   \sqrt{\epsilon^2  + 4 |\gamma|} -
  |\epsilon|}{2|\gamma|} \geq 1 \iff
   \sqrt{\epsilon^2  + 4 |\gamma|} \geq 2|\gamma|
  + |\epsilon| \iff \\
 & \epsilon^2  + 4 |\gamma \geq (2|\gamma|
  + |\epsilon|)^2 \iff  1 \geq |\gamma| +
  |\epsilon| \iff\\
  &1 \geq |\gamma| + |\cos(2p)(1 - |\gamma|)| \iff
  \\
  &1 - |\cos(2p)| \geq (1 - |\cos(2p)|) |\gamma|
  \iff 1 \geq |\gamma|
\end{align*}
which proves  that $|r_\pm| \geq 1$ for $\gamma < 0$.

\subsection{$\epsilon^2
  - 4\gamma < 0$ }
If $\epsilon^2
- 4\gamma < 0$ then
\begin{align*}
  r_\pm =
  \frac{\epsilon \pm i \sqrt{4\gamma -
  \epsilon^2}}
  {2\gamma}
\end{align*}
which implies $ |r_{\pm}|^2 = \frac1\gamma \geq 1$.

\end{document}